\documentclass[letterpaper, 10 pt, conference]{ieeeconf}  

\IEEEoverridecommandlockouts                              

\usepackage{amsmath}
\usepackage{amssymb}
\usepackage{color}
\usepackage{graphicx}

\newcommand{\bluetext}[1]{{\leavevmode\color{blue}{#1}}}

\newcommand{\R}{\mathbb{R}}
\renewcommand{\t}[1]{\tilde{#1}}

\newtheorem{remark}{Remark}
\newtheorem{prop}{Proposition}
\newtheorem{lemma}{Lemma}
\newtheorem{corollary}{Corollary}
\newtheorem{definition}{Definition}
\newtheorem{theorem}{Theorem}
\newtheorem{problem}{Problem}

\title{\LARGE \bf
Predicting Viral Evolution from a Single Early Measurement \\
Using the Target Cell Limited Model
}

\author{Rahal Nanayakkara$^{1}$ and Paulo Tabuada$^{1}$
\thanks{*This work is supported by NSF award 2146828.}
\thanks{$^{1}$R. Nanayakkara and P. Tabuada are with the Electrical and Computer Engineering Department, University of California at Los Angeles, Los Angeles,
CA 90095 USA,
        {\tt\small rahaln@ucla.edu, tabuada@ee.ucla.edu}}%
}

\begin{document}

\maketitle
\thispagestyle{empty}
\pagestyle{empty}

\begin{abstract}
Recent advances in diagnostic techniques have enabled the accurate quantification of early-stage viral loads. A key problem of interest is translating these measurements into predictive clinical insights, such as forecasting a patient's onset of infectiousness and peak infection severity. In this work, we address this problem using the Target Cell Limited (TCL) model. Because the host's internal biological states are practically unobservable, predicting the viral trajectory from a single noisy viral load measurement is highly nontrivial. To overcome this, we introduce a novel coordinate transformation that converts the nonlinear viral dynamics into a monotone system. By leveraging monotone systems theory and taking into account invariant subspaces of the transformed system, we derive explicit analytical formulae that establish a strict upper bound on the peak viral load and a guaranteed lower bound on the time to infectiousness using a single early viral observation.
\end{abstract}

\section{INTRODUCTION}

The accurate prediction of viral dynamics within an infected host is a fundamental challenge in epidemiology, whose resolution would enable the development of targeted therapeutic interventions. 
Recent advances in viral diagnostic techniques have enabled the quantification of even minute viral loads, moving beyond traditional binary outcome based tests \cite{artika2022rtpcr, suo2020ddpcr}.
Consequently, quantitative measurements obtained during the early stages of infection can now be actively utilized in predicting viral evolution.

Furthermore, the global impact of the COVID-19 pandemic has created a surge of interest in using detailed, patient-specific information to improve epidemic modelling and design personalized intervention strategies \cite{hernandez2020host, hernandez2022modelling_epidemics}. 
Under this shift in paradigm, a new problem emerges: How to leverage mathematical models to translate these early, isolated viral measurements into reliable predictive clinical insights, based on detailed patient information.
Of particular interest is the accurate forecast of key milestones of the viral trajectory, such as the exact onset of patient infectiousness and the ultimate severity of the infection.

To address these questions, we investigate the Target Cell Limited (TCL) model for acute infections \cite{ciupe2017host}, commonly referred to as the UIV model as it captures three in-host states: uninfected target cells ($U$), infected cells ($I$), and viral load ($V$). 
While macroscopic compartmental frameworks, such as the SIR or SEIR models, characterize the transmission of epidemics across broad populations, the UIV model focuses on in-host dynamics, capturing the underlying mechanisms governing the evolution of an infection within a single host. 
This model also possesses the distinct advantage of having a direct relationship between the internal state $V$ 
and the measured outcomes of clinical tests.

The UIV model has been widely employed in the epidemiological literature to analyse the in-host dynamics of a multitude of diseases, including Dengue \cite{ben2015dengue}, HIV \cite{perelson2013hiv}, Hepatitis \cite{reluga2009hepatitis}, Ebola \cite{nguyen2015ebola}, Influenza \cite{baccam2006influenza}, and, more recently, COVID-19 \cite{hernandez2020host}. 
Furthermore, under an appropriate parameter restriction and relabeling of states, the UIV dynamics coincide with the standard SEIR dynamics, despite the two models describing fundamentally distinct biological phenomena at different scales.
A major challenge in analyzing the SEIR/UIV dynamics is that, unlike the structurally simpler SIR dynamics \cite{turkyilmazoglu2021sir}, their governing differential equations do not admit exact closed-form solutions \cite{weinstein2020analytic}. 
Consequently, analytically determining the peak of the infection-driving state remains a problem of significant interest within the epidemiological literature \cite{piovella2020analytical}.

To address this, various approximations have been proposed, such as \cite{heng2020approximately}, which shows that SEIR trajectories share a universal shape modified only by spatio-temporal scaling. While \cite{weinstein2020analytic} provides a continuous analytical approximant, computing the peak still necessitates numerical root-finding. Conversely, \cite{piovella2020analytical} derives a closed form formula for the peak via an adiabatic approximation; however, its accuracy degrades unless the basic reproduction number is near 1. 
Additionally, this method also lacks strict bounding guarantees and requires full knowledge of the initial condition.

To overcome these limitations, this paper presents explicit analytical formulas to establish a strict upper bound on the peak viral load $V$, and a guaranteed lower bound for the earliest possible onset of patient infectiousness. 
We achieve this by introducing a novel change of coordinates that transforms the standard UIV model into a monotone dynamical system. 
By leveraging the established theory of monotone systems and analyzing the invariant subspaces of this transformed space, we systematically bound the state trajectories, yielding the required bounds directly in terms of the initial noisy observation.

\section{BACKGROUND} \label{sct:background}

\subsection{Notation and preliminaries}

\textit{Notation :}
We denote by $\R, \R_{\geq0}, \R_{>0}$ the set of real numbers, non-negative real numbers and positive real numbers respectively. $\R^n$ denotes the $n$ dimensional Euclidean space, and $\R^n_{\geq0}$ denotes the positive orthant in $\R^n$, i.e., the set $\{(x_1,\dots,x_n) \in \R^n: x_i \geq 0, i=1,\dots,n \}$.

Given a nonlinear system of the form $\dot x = f(x,w)$, where $x \in \mathcal X \subseteq \R^n$  and $w \in \mathcal W \subseteq \R^m$, an initial condition $x_0 \in \mathcal X$ at time $t=0$, and an input signal \mbox{$\mathbf{w} : [0, \infty) \to \mathcal W$}, we denote by $\phi(t, x_0, \mathbf{w})$ its solution at time $t \in \R_{\geq 0}$.

For any $a,b \in \R^n$, we say $a \leq b$ if $a_i \leq b_i$ for all \mbox{$i=1,\dots,n$} where $a_i$ and $b_i$ represent the $i$\textsuperscript{th} coordinate of $a$ and $b$ respectively. 
A set $C \subseteq \R^n$ is said to be order-convex if for every $a,b \in C$ with $a \geq b$, and every $\lambda \in [0, 1]$, it holds that $\lambda a + (1-\lambda)b \in C$. 
For instance, any ordered convex set is also order-convex.
Next we define the notion of a monotone system with respect to this ordering relation.

\begin{definition} (Monotone System \cite{angeli2003monotone}) 
A system \mbox{$\dot x = f(x, w)$} is monotone if for any pair $x, y \in \mathcal X$ such that $x \leq y$, and any pair $\mathbf{w_1, w_2} : [0, \infty) \to \mathcal W$ such that $\mathbf{w_1}(t) \leq \mathbf{w_2}(t)$ for all $t \geq 0$, it holds that $\phi(t, x, \mathbf{w_1}) \leq \phi(t, y,\mathbf{w_2})$ for all $t \geq 0$.
\label{def:monotone_sys}    
\end{definition}

We recall the following necessary and sufficient condition for monotonicity of a control system.

\begin{theorem}(Proposition III of \cite{angeli2003monotone})
    A control system \mbox{$\dot x = f(x,w)$}, where \mbox{$x \in \mathcal X \subset \R^n$}, $w \in \mathcal W \subseteq \R^m$, $\mathcal{X}, \mathcal{W}$ are order-convex, and $f$ is continuously differentiable, is monotone if and only if:
    \begin{align}
        \frac{\partial f_i}{\partial x_j} &\geq 0, \quad \forall x \in \mathcal X, w \in \mathcal W, \quad \forall i \neq j, \\
        \frac{\partial f_i}{\partial w_j} &\geq 0, \quad \forall x \in \mathcal X, w \in \mathcal W, \quad \forall i,j.
    \end{align}
    \label{thm:monotone_cond}
\end{theorem}

We say a set $S \subseteq \mathcal X$ is forward invariant under some dynamics $\dot x = f(x)$ if
$x_0 \in S $ implies that $\phi(t,x_0) \in S$ for all $ t \geq 0$.
We also recall that if $S$ can be represented as
\mbox{$S = \{x \in \mathcal X : h(x) \geq 0 \}$}, where $h : \mathcal X \to \R$ is a continuously differentiable function such that $\nabla h(x) \neq 0$ when $h(x)=0$, then $S$ is forward invariant if $\dot h (x) = \nabla h (x) f(x) \geq 0$ for all $x \in \partial S$ \cite{ames2016control}. Here $\partial S$ denotes the boundary of $S$, i.e., the set $\{ x \in \mathcal X : h(x) = 0\}$.

\subsection{The UIV model}

We focus our analysis on the UIV model, also called the Target Cell Limited (TCL) model, for acute infections \cite{ciupe2017host, hernandez2020host} which consists of three in-host states : $U$ - Uninfected Cells, $I$ - Infected Cells, and $V$ - Viral Load. 
The dynamics of this model are given by:
\begin{equation}
    \begin{pmatrix}
        \dot U \\ \dot I \\ \dot V
    \end{pmatrix} =
    \begin{pmatrix}
        - \beta U V \\
    \beta U V - \delta I \\
    p I - c V
    \end{pmatrix},
    \label{eq:uiv}
\end{equation}
where $\beta, \delta, p, c \in \R_{>0}$ are model parameters.
The parameter $\beta$ represents the infection rate of healthy cells, $\delta$ the death rate of infected cells, $p$ the replication rate of the virus, and $c$ the clearance rate of the virus (the combined effect of both the immune system and natural death of viral particles) \cite{hernandez2020host}.
The UIV model belongs to a broader class of within-host viral dynamics models which aim to capture the mechanisms driving infections on a microscopic scale \cite{ciupe2017host}.

\begin{remark}
    Under the parameter restriction $\delta=p$ and the state mapping $(U,I,V)\mapsto(S,E,I)$, the dynamical system \eqref{eq:uiv} coincides with the normalized SEIR model \cite{cooke1967seir_orig, li1995seir_global}, with the additional (redundant) state \mbox{$R=1-(S+E+I)$}. Thus, under these conditions, all results derived in this paper are also applicable to the SEIR model through the corresponding renaming of variables.
\end{remark}

\begin{remark}
        The positive orthant $\R^3_{\geq 0}$
        is forward invariant under the dynamics \eqref{eq:uiv} \cite{hernandez2022modelling_epidemics}. Since this set contains all possible states of interest, we restrict our analysis to this invariant domain.
        \label{rem:fwd_invariant}
\end{remark}

The state $U$ represents the ``target cells'' for the specific infection under study, for instance in the case of COVID-19, this corresponds to respiratory epithelial cells \cite{hernandez2020host}.
An uninfected host has a state $(U_0, 0, 0)$, where $U_0 \in \R_{>0}$ is the nominal number of cells, which is an equilibrium of \eqref{eq:uiv}. 
When the initial infection occurs the host's state becomes $(U_0, 0, V_0)$, where $V_0 \in \R_{>0}$ denotes the initial viral load \cite{hernandez2022modelling_epidemics}.
However, for the remainder of our analysis, we shall assume a more general initial condition $(U_0, I_0, V_0)$, allowing us to analyse trajectories of the system starting from a time later than the time of initial infection.

Similar to \cite{hernandez2020host}, we also define the quantity:
\begin{equation}
    R_0 = \frac{p \beta U_0}{\delta c},
    \label{eq:R0}
\end{equation}
commonly referred to as the ``reproduction number'' in the epidemiology literature.
It is well known that for $R_0 \leq 1$, the viral load $V$ does not grow, but instead decays (i.e., the uninfected equilibrium is stable) \cite{korobeinikov2004global}. Thus for the remainder of the paper we only consider the case where $R_0>1$, as the other case is not of interest.

\begin{remark}
    The states $I$ and $V$ evolve in such a manner that they each have a single peak at distinct times, and converge to $\lim_{t \to \infty} I(t) = \lim_{t \to \infty} V(t) =0$ \cite{hernandez2022modelling_epidemics}.
    This convergence to zero represents the resolution of the infection, which, depending on other biological factors, corresponds to either recovery from the infection or, unfortunately, death.
\end{remark}

\begin{remark}
    The state $U(t)$ is monotonically decreasing with $t$ and converges to \mbox{$\lim_{t \to \infty} U(t) = U_\infty > 0$} given by:
    \begin{equation*}
        U_\infty = -\frac{U_0}{R_0} W_0 \left( -R_0 e^{-\frac{R_0}{U_0} (U_0 + I_0 + \frac{\delta}{p} V_0)} \right),
    \end{equation*}
    where $W_0 : \R \to \R$ is the principal branch of the Lambert W function \cite{hernandez2022modelling_epidemics}.
\end{remark}

In the field of epidemiology, a host is typically deemed infectious, i.e., the host has the ability to infect other hosts, when the viral load $V$ is above a certain threshold whose specific value depends on the disease under study \cite{zhou2023viral_thresh, van2021duration_thresh}. The period of time between initial infection and the viral load reaching this threshold is typically referred to as the ``incubation stage''.

\section{PROBLEM FORMULATION} \label{sct:problem_formulation}

While the states $U$ and $I$ are typically not observable, modern testing methods, such as qPCR \cite{artika2022rtpcr} and ddPCR \cite{suo2020ddpcr} in the case of COVID-19, allow us to obtain an estimate of the state $V$ at the time of testing.
In practice, the results of such a test will not be exact, and only a range of possible viral loads $[\underline{V}_s,\overline{V}_s]$ may be inferred.
In this paper, we utilize the model \eqref{eq:uiv} to address the following problems based on the outcome of such a test:

\begin{problem}
Given a single uncertain measurement of the viral load prior to its peak, determine an upper bound for the peak viral load attained during the course of the infection.
\label{prob:peak}
\end{problem}

\begin{problem}
Given a single uncertain measurement of the viral load obtained during the ``incubation stage,'' determine a lower bound for the time required for the patient to become infectious.
\label{prob:time}
\end{problem}

The ability to estimate these critical quantities from a single test enables healthcare professionals to tailor care to individual patient needs. Specifically, bounding the peak viral load allows clinicians to anticipate disease severity and proactively explore appropriate therapeutic interventions. Furthermore, previous research \cite{abuin2021antiviral} demonstrates that the efficacy of certain medications can be incorporated into the model as parameter adjustments. By re-evaluating our proposed bounds under these modified parameters, clinicians can better predict a patient's response to various treatment regimens.

Addressing Problem \ref{prob:time} holds equal practical importance, as it essentially determines the earliest time at which a host becomes infectious. 
Determining this lower bound for the viral incubation period facilitates more informed public health decisions, and such information is critical in resource-constrained scenarios where isolation capacity must be optimized.

When proposing our solutions to these problems, we assume that we have full knowledge of the model parameters $\beta, \delta, p$ and $c$. In practice, accurate estimates of these parameters can be obtained using past clinical data \cite{alamo2021data, hernandez2019modelingmatlab}. Furthermore, in scenarios where only a range of possible values is available, the final expressions for the derived bounds can be appropriately substituted with the extreme values of the parameters to obtain worst case guarantees.

\section{MONOTONICITY} \label{sct:monotonicity}

\subsection{Change of Variables}

We first note that it is possible to solve the first equation of \eqref{eq:uiv}, to obtain:
\begin{equation*}
    U(t) = U_0 e^{-\beta Z(t)},
\end{equation*}
where:
\begin{equation*}
    Z(t) = \int_0^t V(\tau) d \tau.
\end{equation*}
This motivates us to define the following change of variables:
\begin{equation}
    (Z, I, V) = \left(\frac{-1}{\beta}\ln\left(\frac{U}{U_0}\right),I, V\right),
    \label{eq:change_of_var}
\end{equation}
which is a diffeomorphism within 
our domain of interest;
\mbox{$\{ (U,I,V) : U_0 \geq U \geq U_\infty > 0, I \geq 0, V \geq 0\}$}.
The dynamics \eqref{eq:uiv}, can now be rewritten as:
\begin{equation}
    \begin{pmatrix}
        \dot Z \\ \dot I \\ \dot V
    \end{pmatrix} =
    \begin{pmatrix}
        V \\
    \beta U_0 V e^{-\beta Z} - \delta I \\
    p I - c V
    \end{pmatrix},
    \label{eq:ziv}
\end{equation}
with initial condition $(0, I_0, V_0)$. 

In practice, the exact value of $U_0$ is typically unknown. Instead, we assume that $U_0$ lies within a known interval of nominal values, $[\underline{U}_0, \overline{U}_0]$. Consequently, the final expressions for the bounds are given in terms of this uncertainty range.

\begin{remark}
    The positive orthant $\R^3_{\geq 0}$
    is forward invariant under the dynamics \eqref{eq:ziv}. Additionally, $Z(t)$ is strictly monotonically increasing with time, since $\dot Z = V > 0$ for all $t \in \R_{\geq0}$.
    \label{rem:fwd_invariant_ziv}
\end{remark}

\begin{prop}
The quantity:
\begin{equation}
    K(t) = U_0 e^{-\beta Z(t)} + \frac{\beta U_0}{R_0}Z(t) + I(t) + \frac{\delta}{p}V(t),
    \label{eq:const_ziv}
\end{equation}
is invariant under the dynamics \eqref{eq:ziv}.
\end{prop}
\begin{proof}
    Direct computation shows that $\frac{dK}{dt}=0$.
\end{proof}

Since the quantity \eqref{eq:const_ziv} is a constant of motion, for the remainder of this paper, we omit the dependence on time and simply refer to it as $K$. The value of this constant can be expressed in terms of the initial conditions as:
\begin{equation}
    K = U_0 + I_0 + \frac{\delta}{p}V_0.
    \label{eq:const_val}
\end{equation}

Our end goal is to derive bounds for the quantities of interest associated with the solutions of \eqref{eq:uiv}. 
We achieve this by leveraging the theory of monotone systems. To this end, we first define the following control system:
\begin{equation}
    \begin{pmatrix}
        \dot {\t Z} \\ \dot {\t I} \\ \dot {\t V}
    \end{pmatrix} =
    \begin{pmatrix}
        \t V \\
    \beta U_0 \t V w - \delta \t I \\
    p \t I - c \t V
    \end{pmatrix},
    \label{eq:ziv_w}
\end{equation}
where $w \in \R_{\geq 0}$ is a control input. 

\begin{prop}
    System \eqref{eq:ziv_w} is monotone for all $w\geq 0$.
    \label{prop:ziv_monotone}
\end{prop}
\begin{proof}
Letting $x = (\t Z,\t I,\t V)$, we write the system in the form $\dot x = f(x,w)$, and compute:
\begin{equation*}
    \nabla_x f = 
    \begin{pmatrix}
        0 & 0 & 1 \\
        0 & -\delta & \beta U_0 w \\
        0 & p & -c\\
    \end{pmatrix} 
    \quad \text{and} \quad \nabla_w f = 
    \begin{pmatrix}
        0 \\ \beta U_0 \t V \\ 0
    \end{pmatrix}.
\end{equation*}
Given that $w(t) \geq 0$ for all $t\geq0$, we have that $\frac{\partial f_i}{\partial x_j} \geq 0$ for all $i,j = 1,2,3$ with $i \neq j$, and $\frac{\partial f_i}{\partial w} \geq 0$ for all $i=1,2,3$. 
Thus by Theorem \ref{thm:monotone_cond}, \eqref{eq:ziv_w} is monotone for any $w(t) \geq 0$.
\end{proof}

For the remainder of our analysis we fix $w(t)$ to be a constant parameter for all $t$, i.e., $w(t)=w \in \R_{>0}$, reducing (\ref{eq:ziv_w}) to a linear dynamical system:
\begin{equation}
    \begin{pmatrix}
        \dot {\t Z} \\ \dot {\t I} \\ \dot {\t V}
    \end{pmatrix} = 
    \underbrace{
    \begin{pmatrix}
        0 & 0 & 1 \\
        0 & -\delta & \beta U_0w\\
        0 & p & -c\\
    \end{pmatrix}
    }_{A(w)}
    \begin{pmatrix}
        \t Z \\ \t I \\ \t V
    \end{pmatrix}.
    \label{eq:lin_ziv_w}
\end{equation}

\begin{theorem}
    Given an initial condition $x_0 = (Z_0, I_0, V_0)$ at time $t_0 \in \R_{\geq0}$, 
    it holds that $\phi(t, x_0) \leq \t \phi(t, x_0)$ for all \mbox{$t \geq t_0$}, where
    $\t \phi(t, x_0)$ denotes the solutions of \eqref{eq:ziv}, and
    $\phi(t,x_0)$ denotes the solutions of \eqref{eq:lin_ziv_w} with $w = e^{-\beta Z_0}$. 
    \label{thm:ziv_bound}
\end{theorem}
\begin{proof}
    Since $Z(t)$ is monotonically increasing, it follows that $e^{-\beta Z(t)} \leq e^{-\beta Z_0}$ for all $t \geq t_0$.
    From monotonicity of the system \eqref{eq:ziv_w}, established in Proposition \ref{prop:ziv_monotone} and Definition \ref{def:monotone_sys}, it follows that the solutions of \eqref{eq:ziv_w} with $w=e^{-\beta Z_0}$ upper bound those of \eqref{eq:ziv_w} with $w = e^{-\beta Z(t)}$ for all $t \geq t_0$.
    But \eqref{eq:ziv_w} with $w=e^{-\beta Z(t)}$ is simply the system \eqref{eq:ziv}, and the result follows.
\end{proof}
The preceding theorem allows us to analyze the behaviour of a linear system initialized at a point on the trajectory of our original nonlinear system, to obtain upper bounds on the state variables. We shall informally refer to this linear system as the linear upper bounding system.

\begin{prop}
    The quantity:
    \begin{equation}
        \t K(w, t) = \beta U_0 \left(\frac{1}{R_0} - w \right)\t Z (t) + \t I (t) + \frac{\delta}{p} \t V (t),
        \label{eq:const_lin}
    \end{equation}
    parameterized by $w$ is invariant under the dynamics (\ref{eq:lin_ziv_w}).
\end{prop}
\begin{proof}
    Direct computation shows that $\frac{d \t K}{dt}=0$.
\end{proof}

As before, since \eqref{eq:const_lin} is a constant of motion, we omit the dependance on time and simply denote it by $\t K(w)$. Additionally, if the initial condition of \eqref{eq:lin_ziv_w} lies on the solution trajectory of \eqref{eq:ziv}, the constants of motion $K$ and $\t K (w)$ are related by the following proposition.

\begin{prop}
    Let $\t \phi$ be a solution of \eqref{eq:lin_ziv_w} and $\phi$ be a solution of \eqref{eq:ziv}. If there exists $t_1, t_2 \in \R_{\geq 0}$ such that \mbox{$\phi (t_1) = \t \phi (t_2)$} then,
    the constant of motion $K$ is related to $\t K(w)$ by:
    \begin{equation}
        \t K(w) = K - w \beta U_0 Z(t_1) - U_0 e^{-\beta Z(t_1)},
        \label{eq:const_rel}
    \end{equation}
    where $(Z(t_1), I(t_1), V(t_1)) = \phi (t_1) = \t \phi (t_2).$
    \label{prop:const_relation}
\end{prop}
\begin{proof} Since $(Z(t_1), I(t_1), V(t_1))$ lies on the trajectory of both systems, subtracting (\ref{eq:const_ziv}) from (\ref{eq:const_lin}), both evaluated at $(Z(t_1), I(t_1), V(t_1))$ yields the relationship.
\end{proof}

The system matrix $A(w)$ of \eqref{eq:lin_ziv_w} has eigenvalues $0, \lambda_-(w)$ and $\lambda_+(w)$, where:
\begin{equation}
    \lambda_\pm(w) = \frac{-1}{2}\left(\delta + c \pm \sqrt{(\delta+c)^2 + 4 \delta c(R_0 w - 1)} \right).
    \label{eq:eigval}
\end{equation}

In the remainder of the paper, for the sake of brevity, we simply denote by $\lambda_-$ and $ \lambda_+$ the eigenvalues corresponding to the case when $w=1$.

\section{MAIN RESULTS} \label{sct:main_results}

\subsection{Bounding Internal States}

Since we aim to provide estimates based solely on a single sample of $V$, we first derive a bound on the remaining state variables in terms of this observable variable. 

\begin{prop}
    The set:
    \begin{equation*}
        \mathcal T = \{ (Z, I, V): Z \geq 0, I \geq 0, \, \beta U_0 V \geq (\lambda_+ + \delta) I \},
    \end{equation*}
    is forward invariant under the dynamics \eqref{eq:ziv}.
\end{prop}
\begin{proof}
    We first observe that $\mathcal T$ defines the region between three hyperplanes given by $Z=0$, $I=0$ and \mbox{$\beta U_0 V = (\lambda_+ + \delta) I$}.
    Letting $x=(Z,I,V)$, we define $h(x) = \beta U_0 V - (\lambda_+ + \delta) I$.
    Now, $\mathcal T = \R^3_{\geq 0} \cap \mathcal S$, where \mbox{$\mathcal S = \{ x: h(x) \geq 0 \}$}.
    Since $\R^3_{\geq 0}$ is forward invariant, we can conclude that $\mathcal T$ is forward invariant by showing that $\dot{h}(x) \geq 0$ for all $x \in \partial S \cap \R^3_{\geq 0}$.
    Computing:
    \begin{align*}
        \dot h ( x) &= \beta U_0 (p I - c V) - (\lambda_+ + \delta)( \beta U_0 V e^{-\beta Z} - \delta I) \\
        &= (p \beta U_0 + (\lambda_+ + \delta)\delta)I - (c + (\lambda_+ + \delta) e^{-\beta Z}) \beta U_0 V,
    \end{align*}
    and evaluating $\dot h (x)$ on the boundary $\partial \mathcal S$, by substituting \mbox{$\beta U_0 V = (\lambda_+ + \delta) I$}, we obtain:
    \begin{align*}
        \begin{split}
            \left. \dot h(x) \right|_{\partial \mathcal S} &= (p \beta U_0 + (\lambda_+ + \delta)\delta)I \\& \qquad- (c + (\lambda_+ + \delta) e^{-\beta Z})(\lambda_+ + \delta)I
        \end{split} \\
        \begin{split}
            &= (\lambda_++\delta)^2(1-e^{-\beta Z})I \\
            & \qquad \qquad - \underbrace{(\lambda_+^2 + (\delta+c)\lambda_+ + c\delta - p \beta U_0)}_{0}I
        \end{split} \\
        &= (\lambda_++\delta)^2(1-e^{-\beta Z})I,
    \end{align*}
    where the second term in the second equality is zero since $\lambda_+$ satisfies \eqref{eq:eigval}. Since $e^{-\beta Z} \leq 1$ for all $Z \geq 0$, we conclude that $\dot h (x) \geq 0$ for all $x \in \R^3_{\geq 0} \cap \partial \mathcal S$, and hence $\mathcal{T}$ is forward invariant.
\end{proof}

Since the set $\mathcal{T}$ is forward invariant, and the initial condition of the host, $(0, 0, V_0)$, lies within $\mathcal T$ for any $V_0 \in \R_{>0}$,
the state of the system at any time $t$ lies within $\mathcal{T}$. Hence:
\begin{align}
    I(t) & \leq \frac{\beta U_0 V(t)}{\lambda_+ + \delta}.
    \label{eq:bound_i}
\end{align}

\begin{remark}
    In our analysis, we consider a more general initial condition $(0, I_0, V_0)$ of some infected host at some time $t_0$, where $t_0$ occurs after the start of infection. Although any arbitrary such state is not guaranteed to be within $\mathcal T$, from a biological perspective all infected hosts start with some finite viral load and 0 infected cells \cite{hernandez2022modelling_epidemics}, i.e., the host's state at the start of infection lies within $\mathcal T$. Hence, any initial condition $(0, I_0, V_0)$ encountered in a host is also guaranteed to lie within $\mathcal T$ and thus \eqref{eq:bound_i} remains valid.
\end{remark}

To bound the state $Z$ (and consequently $U$) we simply use the fact that $Z(t)$ is monotonically increasing:
\begin{equation}
    Z(t) \geq Z_0 \quad \iff \quad U(t) \leq U_0.
    \label{eq:bound_u}
\end{equation}

\subsection{Peak Viral Load}

In this section we address Problem \ref{prob:peak}, where given a single uncertain measurement of the viral load $V$, obtained prior to the peak, we derive an upper bound for the peak value $V_p$ attained at some time $T_p$.
We achieve this by applying Theorem \ref{thm:ziv_bound} and selecting a strategic point along the trajectory to initialize the linear upper bounding system.

Theorem \ref{thm:ziv_bound} essentially states that given any point $(Z_\tau, I_\tau, V_\tau)$ on the trajectory of the nonlinear system \eqref{eq:ziv} corresponding to some time $\tau$, we have:
\begin{equation}
    \begin{pmatrix}
        Z(t) \\ I(t) \\ V(t)
    \end{pmatrix} \leq \exp(A(e^{-\beta Z_\tau})(t-\tau))
    \begin{pmatrix}
        Z_\tau \\ I_\tau \\ V_\tau
    \end{pmatrix}, \quad \forall t\geq \tau,
    \label{eq:upper_bound}
\end{equation}
where $\exp(\cdot)$ represents the matrix exponential.

To make a strategic choice for $\tau$, we first observe that the linear dynamics \eqref{eq:lin_ziv_w} is independent of the state $\t Z$, allowing us to extract the subsystem:
\begin{equation}
    \begin{pmatrix}
        \dot {\t I} \\ \dot {\t V}
    \end{pmatrix} = 
    \underbrace{
    \begin{pmatrix}
        -\delta & \beta U_0w\\
        p & -c\\
    \end{pmatrix}
    }_{\t A(w)}
    \begin{pmatrix}
        \t I \\ \t V
    \end{pmatrix}.
    \label{eq:lin_iv_w}
\end{equation}

\begin{lemma}
    Given a system \eqref{eq:ziv} and an initial condition $(Z_0, I_0, V_0) \in \R^3_{\geq 0}$, there exists a unique time $T_0 \in \R_{\geq 0}$ such that the system \eqref{eq:lin_iv_w} with $w=e^{-\beta Z(\tau)}$ is:
    \begin{enumerate}
        \item unstable for $\tau < T_0$.
        \item marginally stable when $\tau = T_0$.
        \item exponentially stable for $\tau > T_0$.
    \end{enumerate}
    Furthermore:
    \begin{equation}
    Z(T_0) = \frac{1}{\beta} \ln (R_0).
    \label{eq:Z_T0}
\end{equation}
\label{lem:T0}
\end{lemma}

\begin{proof}
    Given the solutions of \eqref{eq:ziv} from an initial condition $(Z_0, I_0, V_0) \in \R^3_{\geq 0}$,
    for any time $\tau \in\R_{\geq 0}$, the corresponding subsystem of the linear upper bounding system at this point can be constructed with $w = e^{-\beta Z(\tau)}$ in \eqref{eq:lin_iv_w}. Since eigenvalues of $\t A(w)$ are $\lambda_+(w)$ and $\lambda_-(w)$ from \eqref{eq:eigval}, we see that when $e^{-\beta Z(\tau)} < R_0^{-1}$ the system \eqref{eq:lin_iv_w} is exponentially stable and when $e^{-\beta Z(\tau)} > R_0^{-1}$ it is unstable. The rest of the proof follows from the fact that $Z(\tau)$ is strictly monotonically increasing in $\tau$.
\end{proof}

The preceding Lemma states that there exists a unique time $T_0$, when the linear upper bounding system for the dynamics of $I$ and $V$ given by \eqref{eq:lin_ziv_w} switch from being exponentially increasing to exponentially decaying. We now use this specific value of time $\tau = T_0$ to initialize the upper bounding system as in \eqref{eq:upper_bound} and obtain bounds for $V(t)$.

In the remainder of this section, we show that the peak value of $\t V(t)$ of this marginally stable linear system is an upper bound for the true peak $V_p$. Since the result of Theorem \ref{thm:ziv_bound} applied to this system only holds for all $t \geq T_0$, we first need to show that the peak time $T_p$ of $V(t)$ occurs after $T_0$.

\begin{prop}
    Given a system \eqref{eq:ziv} and an initial condition $(Z_0, I_0, V_0)$, the times $T_p$ and $T_0$ satisfy $T_0 < T_p$.
    \label{prop:T_inequality}
\end{prop}
\begin{proof}
    At the peak of $V(t)$, we have that $\dot{V} (T_p)=0$ and $\ddot{V}(T_p) < 0$. Denoting the state at this time by $(Z_p, I_p, V_p)$, the first derivative condition results in:
    \begin{equation}
        p I_p = c V_p,
        \label{eq:v_peak_cond}
    \end{equation}
    and the second derivative condition yields:
    \begin{equation}
        \ddot{V}(T_p) = p \dot I (T_p) - c \dot{V} (T_p) = p \dot I (T_p) < 0.
    \end{equation}
    Substituting for $\dot I$ from the dynamics, we get:
    \begin{equation*}
        \beta U_0 V_p e^{-\beta Z_p} < \delta I_p.
    \end{equation*}
    Substituting from \eqref{eq:v_peak_cond} we have:
    \begin{equation*}
        e^{-\beta Z(T_p)} < \frac{\delta c}{p \beta U_0} = R_0^{-1} = e^{-\beta Z(T_0)}.
    \end{equation*}
    Thus by monotonicity of $Z(t)$ we have $T_0 < T_p$.
\end{proof}

Now, Theorem \ref{thm:ziv_bound} allows us to directly conclude that the peak of the marginally stable linear system, denoted by $\t V_p$, is an upper bound for $V_p$.
What remains is to compute $\t V_p$ in terms of known quantities.
We achieve this by first substituting $w=R_0^{-1}$ in \eqref{eq:const_lin}, to get the constant of motion for our chosen linear system:
\begin{equation}
    \t K(R_0^{-1}) = \t I (t) + \frac{\delta}{p} \t V (t).
    \label{eq:i_v_const}
\end{equation}
Now suppose that $\t V(t)$ attains its peak at some $t=\t T_p$ and denote its state at this point by $(\t Z_p, \t I_p, \t V_p)$. Since at the peak $\dot{\t V} = 0$, we have that:
\begin{equation*}
    p \t I_p - c \t V_p = 0.
\end{equation*}
Combined with \eqref{eq:i_v_const} we get:
\begin{equation}
    \t V_p = \frac{p}{\delta+c} \t K(R_0^{-1}).
    \label{eq:v_tilde_peak}
\end{equation}

We summarize the arguments made in this section in the following theorem and express \eqref{eq:v_tilde_peak} in terms of the initial conditions of the original system.

\begin{theorem}
    Let $(U_0, I_0, V_0)$ be the state of the system \eqref{eq:uiv} at some time $t_0 \leq T_0$.
    The peak value $V_p$ attained by $V(t)$ at time $T_p$ is bounded by:
    \begin{equation}
        V_p \leq \frac{p}{\delta+c} \left[  U_0 \left(1 - \frac{1+\ln (R_0)}{R_0}  \right) +I_0 + \frac{\delta}{p} V_0\right].
        \label{eq:peak_initial_cond}
    \end{equation}
    \label{thm:peak_v}
\end{theorem}

\begin{proof}
    Given a system \eqref{eq:uiv} and the initial condition $(U_0, I_0, V_0)$, we apply the change of variables \eqref{eq:change_of_var} to construct a system of the form \eqref{eq:ziv} with initial state $(0, I_0, V_0)$. Now by Lemma \ref{lem:T0} there exists a unique time $T_0$ such that the subsystem \eqref{eq:lin_iv_w} of the linear upper bounding system with $w = e^{-\beta Z(T_0)} = R_0^{-1}$ is marginally stable. Furthermore by Theorem \ref{thm:ziv_bound}:
    \begin{equation*}
        V(t) \leq \t V(t), \quad \forall t \geq T_0,
    \end{equation*}
    where $\t V(t)$ is the viral load trajectory of the marginally stable linear system, when it is initialized with $(Z(T_0), I(T_0), V(T_0))$.
    Since $T_0 < T_p$, by Proposition \ref{prop:T_inequality}, we have that:
    \begin{equation*}
        V(T_p) \leq \t V_p = \frac{p}{\delta+c} \t K(R_0^{-1}).
    \end{equation*}
    Now, Proposition \ref{prop:const_relation} relates $\t K (R_0^{-1})$ and $K$, yielding:
    \begin{equation*}
    \t V_p = \frac{p}{\delta+c} \left[  K - \frac{U_0}{R_0}(1+\ln (R_0) ) \right].
    \end{equation*}
    Finally, substituting for $K$ in terms of the initial conditions from \eqref{eq:const_val}, we obtain \eqref{eq:peak_initial_cond}.
\end{proof}

Theorem \ref{thm:peak_v} provides an upper bound for the peak value of the viral load $V_p$ in terms of the initial conditions $(U_0, I_0, V_0)$ of the nonlinear system. We now extend this bound to the case where only a single uncertain measurement of the state $V(t)$ is known.

\begin{corollary}
Suppose a host with dynamics \eqref{eq:uiv} is tested at some time $t_s\leq T_0$ 
and the test outcome is $[\underline{V}_s,\overline{V}_s] \ni V(t_s)$.
Then, the peak value $V_p$ attained by $V(t)$ is bounded by:
\begin{equation} \label{eq:v_peak_test_bound}
        V_p \leq \frac{p}{\delta+c} \left[  \overline{U_0} \left(1 - \frac{1+\ln (R_0)}{R_0}  \right) + \left(\frac{\beta \overline{U_0}}{\lambda_+ + \delta} + \frac{\delta}{p}\right)\overline{V}_s\right].
\end{equation}
\end{corollary}
\begin{proof}
Suppose the true state at time $t_s$ is $(U_s, I_s, V_s)$.
Invoking Theorem \ref{thm:peak_v} with this initial condition, we get:
\begin{equation*}
    V_p \leq \frac{p}{\delta+c} \left[  U_s \left(1 - \frac{1+\ln (R_0)}{R_0}  \right) +I_s + \frac{\delta}{p}V_s\right].
\end{equation*}

Now, using \eqref{eq:bound_i} and \eqref{eq:bound_u} to bound $I_s$ and $U_s$, and the fact that \mbox{$1+\ln (R_0) \leq R_0$}, we get \eqref{eq:v_peak_test_bound}.
\end{proof}

The expression \eqref{eq:v_peak_test_bound} answers Problem \ref{prob:peak}, as it provides an upper bound for the peak of the viral load solely in terms of the outcome of a single uncertain test result $[\underline{V}_s, \overline{V}_s]$.

\subsection{Time to infectiousness}

We now address Problem \ref{prob:time}, by lower bounding the time taken for a host to reach the infectiousness threshold $V_\text{Thresh}$.

\begin{theorem}
Suppose a host with dynamics \eqref{eq:uiv} is tested at some time $t=t_s$ and the test outcome is $[\underline{V}_s,\overline{V}_s] \ni V(t_s)$.
Then the time $T$ at which the host becomes infectious, i.e., $V(T)=V_\text{Thresh}$, is lower bounded as:
\begin{equation}
    T \geq  t_s + \frac{1}{\lambda_+} \ln\left( \frac{V_\text{Thresh}}{\overline V_s} \right).
    \label{eq:time_thresh}
\end{equation}
\end{theorem}
\bluetext{}

\begin{proof}
    Substituting \eqref{eq:bound_i} for $I$ in the dynamics of $V$ in \eqref{eq:uiv}, we get:
    \begin{equation*}
        \dot{V} \leq \underbrace{\left(p\frac{\beta U_0}{\lambda_+ + \delta} - c \right)}_{\alpha} V.
    \end{equation*}
    Substituting for $\lambda_+$ from \eqref{eq:eigval}, we get:
    \begin{align*}
        \alpha &= \frac{2p \beta U_0}{\delta - c + \sqrt{(\delta+c)^2 + 4 \delta c(R_0- 1)}} -c \\
        &= \frac{2p \beta U_0 (\delta - c - \sqrt{(\delta+c)^2 + 4 \delta c(R_0- 1)})}{-4 \delta c R_0} - c\\
        &= \frac{-1}{2} \left(\delta+c - \sqrt{(\delta+c)^2 + 4 \delta c(R_0- 1)} \right) = \lambda_+.
    \end{align*}
    Where the second equality comes from rationalizing the denominator and the third from the definition of $R_0$ \eqref{eq:R0}. Hence $\dot V \leq \lambda_+ V$, and from the comparison lemma \cite{khalil2002nonlinear}, it follows that, $V_{\text{Thresh}} \leq \overline{V_s} e^{\lambda_+ (T - t_s)}.$
\end{proof}

\section{SIMULATIONS} \label{sct:simulations}

In this section, we verify our analytical results through numerical simulations of the UIV model applied to COVID-19. For all simulations, we sample parameters from the ranges established in \cite{hernandez2020host}, strictly maintaining an $R_0$ greater than 1. 
According to \cite{hernandez2020host}, the typical value of uninfected cells for COVID-19 is around $4 \times 10^8$ cells, while the initial viral load is around $50 - 100$ copies/mL. 
Furthermore, while traditional PCR tests yield only a binary indication of infection, advanced techniques like qPCR and ddPCR provide quantitative measurements of viral load with accuraccies of around $\pm 10,000$ copies/mL and $\pm 100$ copies/mL respectively \cite{suo2020ddpcr, artika2022rtpcr, vasudevan2021ddpcr}.

\begin{figure}
    \centering
    \includegraphics[width=0.8\linewidth]{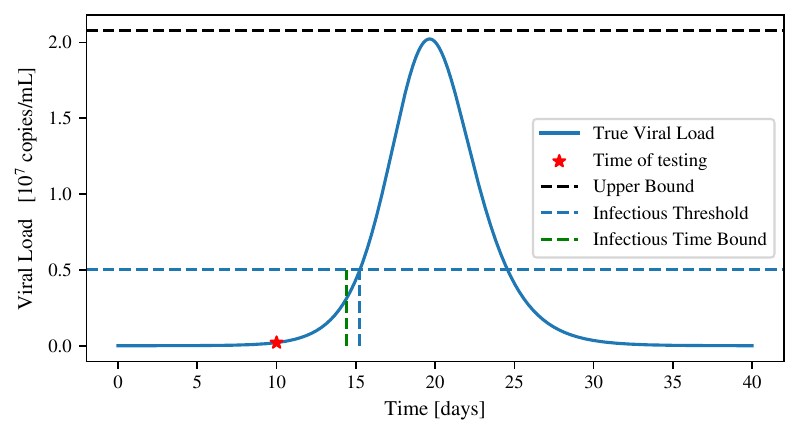}
    \caption{Comparison of true evolution of viral load, the upper bound for the peak and the time to become infectious. }
    \label{fig:viral_traj}
\end{figure}

Fig \ref{fig:viral_traj} illustrates the viral load trajectory for a representative set of parameters ($\beta=1.2\times10^8$, $\delta=1$, $p=1$, and $c=2.3$), assuming an initial viral load of $100$ copies/mL. In this scenario, the host is tested on day $10$, yielding a viral load of $2 \times 10^5$ copies/mL. Accounting for a measurement error of $\pm 1\times10^5$ copies/mL, the upper estimate $\overline{V}_s$ is set to $3 \times10^5$ copies/mL. The dashed black line represents the theoretical upper bound for the peak viral load, calculated via \eqref{eq:v_peak_test_bound}. Furthermore, we validate \eqref{eq:time_thresh} by establishing a lower bound on the time required to reach a viral load of $5 \times 10^6$ copies/mL. This threshold and the actual time at which it is reached are indicated by dashed blue lines, while the theoretical lower bound derived from \eqref{eq:time_thresh} is denoted by the dashed green line.

\begin{figure}
    \centering
    \includegraphics[width=0.75\linewidth]{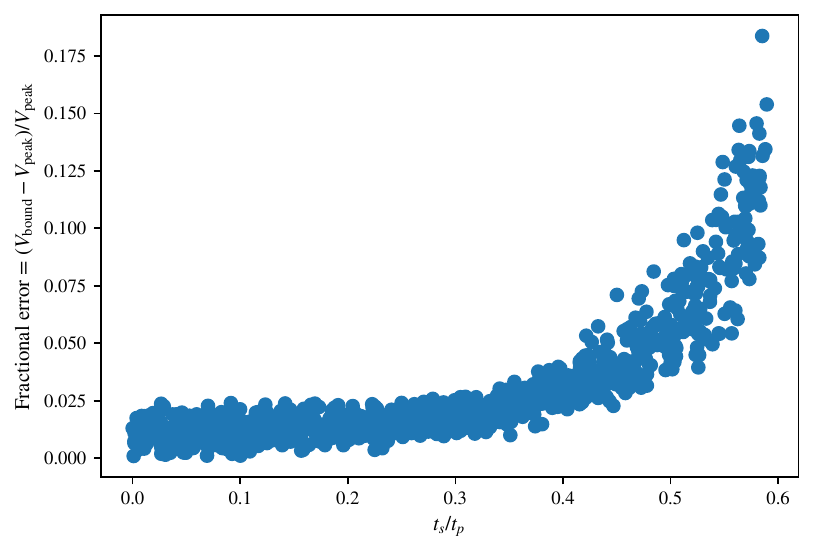}
    \caption{Variation of fractional error of the viral load upper bound with relative time of testing.}
    \label{fig:peak_err_time}
\end{figure}

We also evaluate the tightness of the upper bound derived in \eqref{eq:v_peak_test_bound} as a function of the testing time across 1000 random trials, as shown in Fig \ref{fig:peak_err_time}. 
Since the rate of viral evolution varies across hosts, it is more informative to evaluate the relative time of testing, defined as the testing time $t_s$ normalized by the time of peak viral load $t_p$, rather than the absolute time. Similarly, the difference between the theoretical upper bound and the true peak is normalized by the true peak to yield the fractional error. 
Fig \ref{fig:peak_err_time} clearly indicates that the fractional error of the upper bound increases as the relative time of testing increases, meaning that a tighter upper bound for the peak can be obtained with earlier testing. 
This result is mathematically expected, as the inequalities used to derive the bound, such as \eqref{eq:bound_u} and \eqref{eq:bound_i}, become tighter as $t \to 0$.

\begin{figure}
    \centering
    \includegraphics[width=0.75\linewidth]{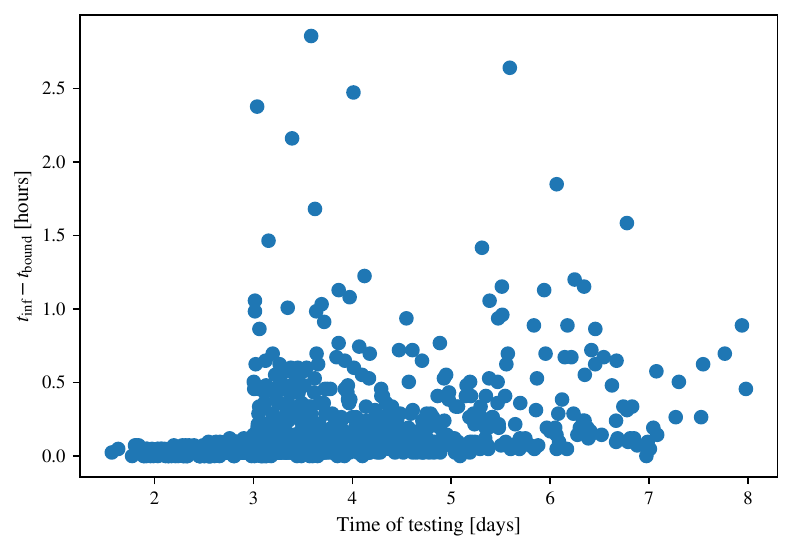}
    \caption{Variation of the difference between true time of becoming infectious and the predicted lower bound, with the time of testing.}
    \label{fig:t_inf_err}
\end{figure}

In Fig \ref{fig:t_inf_err}, we demonstrate the tightness of the bound derived in \eqref{eq:time_thresh} across 1000 trials. 
As formulated in Problem \ref{prob:time}, testing now occurs during the incubation stage, resulting in very small viral load measurements. 
Consequently, we assume measurement accuracies consistent with ddPCR techniques ($\pm 100$ copies/mL). 
The threshold for infectiousness is taken to be $10^5$ copies/mL, based on findings in \cite{zhou2023viral_thresh, van2021duration_thresh}. 
Fig \ref{fig:t_inf_err} reveals that the lower bound is much tighter when tests are administered close to the date of initial infection.
However, even in the worst-case scenario, the error between the theoretical bound and the actual time of infectiousness is around 3 hours, which is considerably tight given that the infection occurs over the course of $30-40$ days.

\section{CONCLUSION}

In this paper, we utilized the Target Cell Limited model (UIV model) to predict in-host viral evolution based on a single uncertain viral load measurement. 
By introducing a novel change of coordinates and using the theory of monotone systems, we established a rigorous upper bound for the patient's peak viral load and a lower bound for the time to reach infectiousness.
We validated our results using numerical simulations based on real COVID-19 patient data.
The simulations led us to conclude that earlier testing led to progressively tighter bounds for both quantities.
These theoretical bounds bridge the gap between mathematically rigorous epidemiological modelling and practical clinical utility, providing robust prognostic guarantees without requiring complete state information.

\addtolength{\textheight}{-12cm}   





\bibliographystyle{IEEEtran}
\bibliography{IEEEabrv,refs}

\end{document}